\documentclass{article}
\usepackage[utf8]{inputenc}
\usepackage{mystyle}

\title{Round and Bipartize for Vertex Cover Approximation}
\author{Danish Kashaev \footnote{\texttt{danish.kashaev@cwi.nl}, Centrum Wiskunde \& Informatica, Amsterdam.}  \and Guido Schäfer \footnote{\texttt{g.schaefer@cwi.nl}, Centrum Wiskunde \& Informatica and ILLC, University of Amsterdam.}}

\date{\vspace{-3ex}}

\begin{document}

\maketitle

\begin{abstract}
The vertex cover problem is a fundamental and widely studied combinatorial optimization problem. It is known that its standard linear programming relaxation is integral for bipartite graphs and half-integral for general graphs. 
As a consequence, the natural rounding algorithm based on this relaxation 
 computes an optimal solution for bipartite graphs and a $2$-approximation for general graphs. 
This raises the question of whether one can interpolate the rounding curve of the standard linear programming relaxation in a beyond the worst-case manner, depending on how close the graph is to being bipartite. In this paper, we consider a simple rounding algorithm that exploits the knowledge of an induced bipartite subgraph to attain improved approximation ratios. Equivalently, we suppose that we work with a pair $(\G, S)$, consisting of a graph with an odd cycle transversal.

If $S$ is a stable set, we prove a tight approximation ratio of $1 + 1/\rho$, where $2\rho -1$ denotes the odd girth (i.e., length of the shortest odd cycle) of the contracted graph $\tilde{\G} := \G /S$ and satisfies $\rho \in [2,\infty]$. If $S$ is an arbitrary set, we prove a tight approximation ratio of $\left(1+1/\rho \right) (1 - \alpha) + 2 \alpha$, where $\alpha \in [0,1]$ is a natural parameter measuring the quality of the set $S$. The technique used to prove tight improved approximation ratios relies on a structural analysis of the contracted graph $\tilde{\G}$, in combination with an understanding of the weight space where the fully half-integral solution is optimal. Tightness is shown by constructing classes of weight functions matching the obtained upper bounds. 

As a byproduct of the structural analysis, we obtain improved tight bounds on the integrality gap and the fractional chromatic number of 3-colorable graphs. We also discuss algorithmic applications in order to find good odd cycle transversals and show that our analysis is optimal in the following sense: the worst case bounds for $\rho$ and $\alpha$, which are $\rho = 2$ and $\alpha = 1 - 4/n$, recover the integrality gap of $2 - 2/n$ of the standard linear programming relaxation, where $n$ is the number of vertices of the graph.

\end{abstract}

\section{Introduction}

In the vertex cover problem we are given a weighted graph $\mathcal{G} = (V, E, w)$, where $w: V \mapsto \mathbb{R}_{+}$ is a non-negative weight function on the vertices, and the goal 
is to find a minimal weight subset of vertices $U \subset V$ that covers every edge of the graph, i.e., 
\[\min \{ w(U) \mid U \subset V,\ |U \cap \{i,j\}| \geq 1\ \forall (i,j) \in E \}.\]
We denote by $OPT$ an optimal subset of vertices for this problem, and by $w(OPT)$ the total weight of that solution. 
The vertex cover problem is known to be NP-complete \cite{karp1972reducibility} and APX-complete \cite{papadimitriou1988optimization}. Moreover, it was shown to be NP-hard to approximate within a factor of $7/6$ in \cite{haastad2001some}, a factor later improved to 1.36 in \cite{dinur2005hardness}. It is in fact NP-hard to approximate within a factor of $2 - \varepsilon$ for any fixed $\varepsilon > 0$ if the unique games conjecture is true \cite{khot2008vertex}. 

A natural linear programming relaxation, as well as its dual, is given by:

\begin{minipage}[t]{0.5\textwidth}
\begin{align*}
    \min \sum_{v \in V} w_v x_v  & \\
    x_u + x_v &\geq 1 \quad \forall (u,v) \in E\\
    x_v &\geq 0 \quad \forall v\in V
\end{align*}
\end{minipage}
\begin{minipage}[t]{0.5\textwidth}
\begin{align*}
    \max \sum_{e \in E} y_e  & \\
    y(\delta(v)) &\leq w_v \quad \forall v \in V\\
    y_e &\geq 0 \quad \; \; \,\forall e\in E
\end{align*}
\end{minipage}

\bigskip
\noindent 
For a given graph $\G$, we denote the primal linear program by $P(\G)$ and the dual by $D(\G)$. 
The integrality gap of the standard linear relaxation $P(\G)$ on a graph of $n$ vertices is upper bounded by $2 - 2/n$, a bound which is attained on the complete graph. In fact, a more fine-grained analysis shows that it is equal to $2 - 2/\chi^f(\G)$, where $\chi^f(\G)$ is the fractional chromatic number of the graph \cite{singh2019integrality}. An integrality gap of $2 - \varepsilon$ is proved for a large class of linear programs in \cite{arora2002proving}. It is also known that any linear program which approximates vertex cover within a factor of $2 - \varepsilon$ requires super-polynomially many inequalities \cite{bazzi2019no}.

An important property of $P(\G)$ is the fact that any extreme point solution $x^*$ is half-integral, i.e., $x^*_v \in \{0, \frac{1}{2}, 1\}$ for any vertex $v \in V$ \cite{nemhauser1975vertex}. 
This gives rise to a straightforward rounding algorithm by solving $P(\G)$ and outputting all vertices whose LP variable is at least a half, i.e., $U := \{v \in V \mid x^*_v \geq \frac12\}$.
It is easy to see that this a $2$-approximation, because $w(U) \leq 2 w(OPT)$, see \cite{hochbaum1982approximation}. Moreover, it is known that $P(\G)$ is integral for any bipartite graph \cite{kuhn1955hungarian}. 
As a consequence, the rounding algorithm returns an optimal solution if the graph is bipartite.
This raises the question of whether we can interpolate the rounding curve of the standard linear program, depending on how close the graph is to being bipartite.

\subsection*{Set-up and algorithm}
We consider the following setup. We are given a weighted non-bipartite graph $\G = (V,E, w)$ and an optimal solution $x^* \in \{0, \frac{1}{2}, 1\}$ to $P(\G)$. We denote by $V_{\alpha}:= \{v \in V \mid x^*_v = \alpha\}$ the vertices taking value $\alpha$ and by $\mathcal{G}_{\alpha} = \mathcal{G}[V_\alpha]$ the subgraph of $\mathcal{G}$ induced by the vertices $V_\alpha$ for any $\alpha \in \{0, \frac{1}{2}, 1\}$. By a standard preprocessing step, we may assume that we only work on the graph $\G_{1/2}$, since any $c$-approximate solution on this reduced graph can be lifted to a $c$-approximate solution on the original graph by adding the nodes $V_1$ to the solution \cite{nemhauser1975vertex}. In addition, we suppose that we have knowledge of an odd cycle transversal $S$ of $\G_{1/2}$, meaning that $\G_{1/2} \setminus S$ is a bipartite graph. Equivalently, $S$ intersects every odd cycle of $\G_{1/2}$.
The question of finding a good such odd cycle transversal is also tackled later in the paper. 

We consider the following simple algorithm, detailed in Algorithm~\ref{round_bipartize}. It first solves $P(\G)$, takes the vertices assigned value one by the linear program to the solution and removes all the integral nodes from the graph to arrive at $\G_{1/2}$. The algorithm then takes all the vertices in the set $S$ to the solution, removes them from the graph and solves another (now integral) linear program to get the optimal solution on the bipartite remainder. These vertices are then also added to the solution.

\begin{algorithm}[]
\caption{}
\label{round_bipartize}
\textbf{Input:} Weighted graph $\mathcal{G} = (V,E,w)$, odd cycle transversal $S \subset V_{1/2}$ \\
\textbf{Output:} Vertex cover $U \subset V$
\begin{algorithmic}[1]
\State Solve the linear program $P(\G)$ to get $V_0 \: , \: V_{1/2}$ and  $V_1$
\State Solve the integral linear program $P(\G_{1/2} \setminus S)$ to get $W \subset V_{1/2}$
\State \Return $V_1 \cup S \cup W$
\end{algorithmic}
\end{algorithm}

The question studied is the following.
What is the worst-case approximation ratio of the algorithm and which weight functions are attaining it?
Our motivation to study this question comes from the structural difference between the polyhedron of $P(\G)$ for bipartite and non-bipartite graphs. 
In particular, we are interested in identifying parameters of the problem that enable us to derive more fine-grained bounds determining the approximation ratio of the algorithm, and allow to interpolate the rounding curve of the standard linear program from 1 to 2, depending on how far the graph is from being bipartite.
As it turns out, the \emph{odd girth}, i.e., the length of the shortest odd cycle, is a key parameter determining tight bounds on the approximation ratio. It is also a natural parameter, since a graph is bipartite if and only if it does contain an odd cycle. The larger the odd girth, the closer the graph is to being bipartite. It is also shown in \cite{gyori1997graphs} that graphs with a large odd girth admit a small cardinality odd cycle transversal. 

\subsection*{Contributions and high-level view}
We first do a pre-processing step and show that we may without loss of generality focus on weighted graphs $\G = (V,E,w)$ where the weights come from a certain weight space $Q^W$. Each edge has a dual weight $y_e \geq 0$ with a total sum of $y(E) = 1$, and the weight on each node is then determined by $w_v = y(\delta(v))$. This follows from the Nemhauser-Trotter theorem, complementary slackness and an appropriate normalization.

We then do the analysis under the assumption that $S$ is a stable set, highlighting the main ideas of the analysis and the proof techniques. We show that the approximation ratio is upper bounded by $1 + 1/\rho$, where $2 \rho -1$ denotes the odd girth of the graph $\tilde{\G} := \G / S$, where all the vertices in $S$ are contracted into a single node. Note that the parameter range is $\rho \in [2,\infty]$, with $\rho = \infty$ naturally corresponding to the case where $\tilde{\G}$ is bipartite.  The proof technique involves a key concept, that we call \emph{pairwise edge-separate} feasible vertex covers. Constructing $k$ such covers allows to bound the approximation ratio by $1+1/k$. The construction of $\rho$ such covers to get the result follows from a structural understanding of the contracted graph $\tilde{\G}$. As a byproduct, this structural understanding also allows to get improved bounds on the integrality gap and the fractional chromatic number of $3$-colorable graphs. In particular, it even manages to compute an exact formula, depending on the odd girth, for the integrality gap and the fractional chromatic number of the contracted graph $\tilde{\G}$.

We then construct a class of weight functions $\mathcal{W} \subset Q^W$ where this upper bound holds with equality, thus showing that this proof technique obtains tight bounds and might have additional applications. This result can then be lifted to the case where $S$ is a general set, by introducing an additional parameter $\alpha$ counting the total dual sum of the weights on the edges inside $S$, i.e. $\alpha = y(E[S])$. This then leads to an approximation ratio interpolating the rounding curve of the standard linear program with a tight bound of $(1 + 1/\rho)(1-\alpha) + 2 \alpha$ for any values of $\rho \in [2,\infty]$ and $\alpha \in [0,1]$.

We then discuss algorithmic applications to find good such sets $S$, connecting to the \Call{MinUnCut}{} and \Call{Colouring}{} problems. Finally, we show that our analysis is optimal in the following sense: the worst case bounds for $\rho$ and $\alpha$, which are $\rho = 2$ and $\alpha = 1 - 4/n$, recover the integrality gap of $2 - 2/n$ of the standard linear programming relaxation for a graph on $n$ vertices.

\subsection*{Implications and related work}

Our analysis falls into the framework of beyond the worst-case analysis \cite{roughgarden2021beyond}. In particular, note that an odd cycle transversal always exists: we may simply take $S = V_{1/2}$, which recovers the standard $2$-approximation algorithm for vertex cover. Depending on how $S$ is chosen, our algorithm can however admit significantly better approximation ratios.

Our algorithm also connects to \emph{learning-augmented algorithms},
which have access to some prediction in their input (e.g., obtained for instance through machine learning). This prediction is assumed to come without any worst-case guarantees, and the goal is then to take advantage of it by making the algorithm perform better when this prediction is good, while still keeping robust worst-case guarantees when it is off \cite{bamas2020primal,lykouris2021competitive,purohit2018improving,lattanzi2020online,antoniadis2020online,antoniadis2020secretary}. In our case, assuming a prediction on the set $S$, robustness is guaranteed since we are never worse than a $2$-approximation. In fact, even if the predicted set is not an odd cycle transversal, one may simply greedily add vertices to it until it becomes one, while still guaranteeing a $2$-approximation. Otherwise, our results provide a precise understanding of how the approximation ratio improves depending on the predicted set $S$. In addition, once such a set $S$ is found, the parameters $\alpha$ and $\rho$ can easily be computed to see the improved guarantee on the approximation ratio. One may thus re-run the machine learning algorithm if the parameters give a bound very close to the worst-case of 2.

The odd cycle transversal number $\text{oct}(\G)$ is defined as the minimum number of vertices needed to be removed in order to make the graph bipartite. The minimum odd cycle transversal problem has been studied in terms of fixed-parameter tractability \cite{reed2004finding,kratsch2014compression}. In particular, it is the first problem where the iterated compression technique has been applied \cite{reed2004finding}, now a classical tool in the design of fixed-parameter tractable algorithms. The best known approximation algorithm for it has a ratio of $O(\sqrt{\log(n)})$ \cite{Agarwal2005OlogNA}. Another relevant concept is the \emph{odd cycle packing number} ocp$(\G)$, defined as the maximum number of vertex-disjoint odd cycles of $\G$ and satisfying ocp$(\G) \leq $ oct$(\G)$. The related maximum stable set problem has been studied in terms of ocp($\G$) in \cite{bock2014solving,artmann2017strongly,conforti2020stable, fiorini2022integer}.

A key property implying the integrality of a polyhedron is the \emph{total unimodularity} (TU) of the constraint matrix describing the underlying problem, meaning that all the square subdeterminants of the matrix are required to lie in $\{-1,0,1\}$ (see for instance \cite{schrijver1998theory,schrijver2003combinatorial}). In general, we believe it is an interesting question to study whether one may exploit the knowledge of a TU substructure in an integer program to obtain improved approximation guarantees through rounding algorithms. In our case, the knowledge of an odd cycle transveral $S$ is equivalent to the knowledge of an induced bipartite subgraph $\G' = \G \setminus S$, for which $P(\G')$ is TU. We hope the techniques introduced for the pair $(\G, S)$ can help tackle other problems.

One technique which might also benefit from our analysis is iterative rounding, which requires solving a linear program at each iteration \cite{lau2011iterative}. Having a better analysis for the case where the linear program becomes integral could potentially be used to reduce the number of iterations and allow for better guarantees, since iterative rounding can terminate at this step without losing solution quality.

Several different algorithms achieving approximation ratios of $2 - o(1)$ have been found for the weighted and unweighted versions of the vertex cover problem: \cite{karakostas2005better,halperin2002improved,bar1983local,monien1985ramsey,bar1981linear,hochbaum1983efficient}. Another large body of work is interested in exact fixed parameter tractable algorithms for the decision version: \cite{buss1993nondeterminism,balasubramanian1998improved,chen2001vertex,chen2000improvement,downey1992fixed,niedermeier1999upper,niedermeier2003efficient,stege1999improved,downey2012parameterized}.

\section{Preliminaries}
\label{section_preliminaries}
We define $\mathbb{R}_+$ to be the non-negative real numbers and $[k] := \{1, \dots, k\}$ to be the natural numbers up to $k \in \mathbb{N}$. For a vector $x \in \mathbb{R}^n$, we denote the sum of the coordinates on a subset by $x(A) := \sum_{i \in A} x_i$.  A key property of $P(\G)$ was introduced by Nemhauser and Trotter in \cite{nemhauser1975vertex}. It essentially allows to reduce an optimal solution $x^* \in \{0,\frac{1}{2},1\}^V$ to a fully half-integral solution by looking at the subgraph induced by the half-integral vertices. As before, we denote by $V_{\alpha}:= \{v \in V \mid x^*_v = \alpha\}$ the vertices taking value $\alpha$ and by $\mathcal{G}_{\alpha} = \mathcal{G}[V_\alpha]$ the subgraph induced by the vertices $V_\alpha$.
\begin{theorem}[Nemhauser, Trotter \cite{nemhauser1975vertex}]
\label{thm_nemhauser1975vertex}
Let $x^* \in \{0,\frac{1}{2},1\}^V$ be an optimal extreme point solution to $P(\G)$. Then,
$w(OPT(\mathcal{G}_{1/2})) = w(OPT(\G)) - w(V_1).$
\end{theorem}

\begin{corollary}
\label{lemma_rounding_type_i}
Let $x^* \in \{0,\frac{1}{2},1\}^V$ be an optimal solution to $P(\G)$. If $S \subset V_{1/2}$ is a feasible vertex cover on $\G_{1/2}$ with approximation ratio at most $\phi$, i.e., $w(S) \leq \phi \: \: w(OPT(\G_{1/2}))$, then
$w(S) +  w(V_1) \leq \phi \: w(OPT(\G)).$
\end{corollary}
\begin{proof}
The proof is an easy consequence of Theorem \ref{thm_nemhauser1975vertex} and the fact that $\phi \geq 1$:
\[w(U) + w(V_1) \leq \phi \: w(OPT(\G_{1/2})) + w(V_1) \leq \phi \: (w(OPT(\G_{1/2})) + w(V_1)) = \phi \: w(OPT(\G)).\]
\end{proof}
The previous corollary thus implies that we may restrict our attention to the graph $\G_{1/2}$, since any $\phi$-approximate solution on this reduced instance can be lifted to a $\phi$-approximate solution on the original graph by adding $V_1$ to the solution. Note that on the weighted graph $\G_{1/2}$, the solution $(\frac{1}{2},\dots,\frac{1}{2})$ is optimal.

For a given set $S \subset V$, we define $\G' := \G \setminus S = (V', E')$ to be the graph obtained by removing the set $S$ and all the incident edges to it. Hence, 
$E = E' \cup \delta(S) \cup E[S]$
where $\delta(S) = \{(u,v) \in E \mid u \in S, v \notin S\}$ and $E[S] := \{(u,v) \in E \mid u \in S, v \in S\}$. We also denote by $\tilde{\G} := \G / S = (\tilde{V}, \tilde{E})$ the graph obtained by contracting all the vertices in $S$ into a single new node $v^S \in \tilde{V}$. We allow for multiple edges, but no self-loops. The only edges present in $E$ but not in $\tilde{E}$ are thus the ones with both endpoints in $S$, i.e., $E[S]$.

\section{Weight Space}

\label{section_weight_space}
By Corollary $\ref{lemma_rounding_type_i}$, we may assume that every weighted graph $\G$ we work with has the property that the fully half-integral solution $x = (\frac{1}{2}, \dots, \frac{1}{2})$ is an optimal solution to the linear program $P(\G)$. In this section, we characterize the weight functions satisfying this assumption. 
\begin{lemma}
\label{lemma_feasible_weights}
Let $\mathcal{G} = (V,E)$ be a graph and let $w: V \to \mathbb{R}_{+}$ be a weight function. The feasible solution $(\frac{1}{2}, \dots, \frac{1}{2})$ to the linear program $P(\G)$ is optimal if and only if there exists $y \in \mathbb{R}^E_{+}$ satisfying $y(\delta(v)) = w_v$ for every $v \in V$.
\end{lemma}
\begin{proof}
By complementary slackness, a feasible solution $x \in \mathbb{R}^V_{+}$ to the primal $P(\mathcal{G})$ and a feasible solution $y \in \mathbb{R}^E_{+}$ to the dual $D(\G)$ are optimal if and only if:
\[x_v > 0 \implies y(\delta(v)) = w_v \quad \forall v \in V\]
and 
\[y_e > 0 \implies x_u + x_v = 1 \quad \forall e = (u,v) \in E.\]
If $x = (\frac{1}{2}, \dots, \frac{1}{2})$ is an optimal solution, then, by strong duality, there exists an optimal dual solution $y$ and this solution needs to satisfy $y(\delta(v)) = w_v$ for every $v \in V$. Conversely, if there exists a dual solution $y$ satisfying $y(\delta(v)) = w_v$ for every $v \in V$, then the pair $x = (\frac{1}{2}, \dots, \frac{1}{2})$ and $y$ satisfy both the conditions of complementary slackness, implying that $x$ is optimal.
\end{proof}

Such instances have been called \emph{edge-induced} in \cite{conforti2020stable, fiorini2022integer}, in the sense that the dual values on the edges are free parameters, and the weights on the nodes are determined once the dual values are fixed. Such instances also satisfy:
\[w(V) = \sum_{v \in V}w_v = \sum_{v \in V}y(\delta(v)) = \sum_{v \in V} \sum_{e \in E}y_e \: 1_{\{e \in \delta(v)\}} = \sum_{e \in E}y_e \sum_{v \in V} 1_{\{e \in \delta(v)\}} = 2 \; y(E).\]
Observe that the approximation ratio of a feasible solution $U \subset V$ is defined as $w(U)/w(OPT(\G))$ and is invariant under scaling of the weights. We thus make a normalization ensuring that the optimal LP solution has objective value one, i.e., $w(V)/2 = y(E) = 1$, to get the following weight space polytope:

\[\Qw := \big\{w \in \mathbb{R}_{+}^V \mid \exists y \in [0,1]^E \text{ such that } y(E) = 1 \text{ and } w_v = y(\delta(v)) \quad \forall v \in V\big\}.\]
It turns out that the vertices of the polytope $Q^W$ have a one-to-one correspondance with the edges of the graph, a proof of which is shown in the Appendix. We end this section by showing that this normalization of the weight space allows us to get a convenient lower bound on $w(OPT(\G))$.

\begin{lemma}
\label{lemma_lower_bound_opt}
Let $\G = (V,E)$ be a graph. For any $w \in Q^W$,
$w(OPT(\G)) \geq 1$.
\end{lemma}
\begin{proof}
Since $w \in Q^W$, we know that the fully half-integral solution is an optimal linear programming solution, showing
$1 = w(V) /2 \leq w(OPT(\G))$,
by feasibility of $OPT(\G)$.
\end{proof}

\section{Round and Bipartize}
\label{sec:round-and-bipartize}

\subsection{Algorithm}
This section is devoted to the analysis of the approximation ratio of the algorithm and is the main contribution of the paper. We assume that we are given as input a pair $(\G,S)$ consisting of a weighted graph and an odd cycle transversal $S \subset V$. By Corollary \ref{lemma_rounding_type_i}, we may assume that the weight function satisfies $w \in Q^W$. By the previous section, there are dual edge weights $y_e \geq 0$ such that $w_v = y(\delta(v))$ for every $v\in V$ and which satisfy $\sum_{e \in E}y_e = 1$.

The algorithm is now very simple. First, take the vertices in $S \subset V$ to the cover and remove them from the graph. Then solve the integral linear program $P(\G \setminus S)$ and take the vertices having LP value one to the cover. The \emph{approximation ratio}, given a weight function $w \in Q^W$, is thus defined as
\begin{equation}
\label{eq_approximation_ratio}
R(w) := \frac{w(S) + w(OPT(\G \setminus S))}{w(OPT(\G))}.
\end{equation}
For simplicity of notation, we omit the dependence on $w$ of $OPT(\G)$ and $OPT(\G \setminus S)$. As a reminder, the bipartite graph $\G \setminus S$ is denoted by $\G' = (V', E')$. The vertex contracted graph $\G / S$ is denoted by $\tilde{\G} = (\tilde{V}, \tilde{E})$ and the contracted node is denoted by $v_S$.

\subsection{Stable Set to Bipartite}
We assume in this section that $S$ is a stable set. We will then generalize the results obtained here in a natural way to the most general setting of an arbitrary set $S$.
We now state our main theorem of this section.

\begin{theorem}
\label{thm_bound_shortest_odd_cycle}
Let $(\G,S)$ be the input to the rounding algorithm, with $S$ being a stable set. For any $w \in \Qw$, the approximation ratio satisfies
\[R(w) \leq 1 + \frac{1}{\rho}\]
where $2 \rho - 1$ is the odd girth of the contracted graph $\mathcal{\tilde{G}}$ and satisfies $\rho \in [2, \infty]$. Moreover, this bound is tight and is attained for a class of weight functions $\mathcal{W} \subset Q^W$.
\end{theorem}
\begin{remark}
We define the odd girth of a bipartite graph as being $\infty$.
\end{remark}
\begin{definition}
Let $(\G, S)$ be a pair consisting of a graph with an odd cycle transversal $S$. For a feasible vertex cover $U \subset V \setminus S$ of the bipartite graph $\G' = \G \setminus S$, we define
\[E_U := \big\{(u,v) \in E \; \big| \; u \in U,\ v \in U \text{ or } \; u \in U,\ v \in S\big\}. \]
In words, these are the edges with either both endpoints in the cover $U$, or with one endpoint in $U$ and one in $S$.
\end{definition}

\begin{definition}
Let $(\G, S)$ be a pair consisting of a graph with an odd cycle transversal $S$. Feasible vertex covers $U_1, \dots, U_k$ of the bipartite graph $\G' = \G \setminus S$ are defined to be \emph{pairwise edge-separate} if the edge sets $\{E_{U_1}, \dots, E_{U_k}\}$ are pairwise disjoint.
\end{definition}

\begin{remark}
We will often say that covers are pairwise edge-separate for the pair $(\G, S)$. It is however worth emphasizing that these covers are defined on the bipartite graph $\G' = \G \setminus S$.
\end{remark}

This definition turns out to be the key concept for us in order to prove improved bounds on the approximation ratio of the algorithm, as shown by the next lemma.

\begin{lemma}
\label{thm_bound_using_pecovers}
Let $(\G,S)$ be the input to the rounding algorithm, with $S$ being a stable set. If there exists $k$ pairwise edge-separate feasible vertex covers of the bipartite graph $\G' = \G \setminus S$, then, for every $w \in Q^W$, the approximation ratio of the algorithm satisfies
\[R(w) \leq 1 + \frac{1}{k}.\]
\end{lemma}
\begin{proof}
Let $w \in \Qw$ and let $y \in \mathbb{R}^E_+$ be the corresponding dual solution satisfying $w_v = y(\delta(v))$ and $y(E) = 1$. We denote by $\{U_1, \dots, U_k\}$ the pairwise edge-separate covers of $\G' = (V', E')$. We can now write down the weights of $S$ and every feasible cover $U_i$ with the help of the dual variables:
\begin{align*}
w(S) &= \sum_{v \in S} w_v = \sum_{v \in S} y(\delta(v)) = y(\delta(S)) \\
w(U_i) &= \sum_{v \in U_i} w_v = \sum_{v \in U_i} y(\delta(v)) = y(E') + y(E_{U_i}) \qquad \forall{i \in [k]}
\end{align*}
The first equality holds because $S$ is a stable set and thus only has edges crossing the set. The second equality holds because every $U_i$ counts the dual value $y_e$ of every $e \in E'$ at least once, by feasibility of the cover, and thus giving a contribution of $y(E')$. The edges in $E_{U_i}$ then give an additional contribution of $y(E_{U_i})$.

By Lemma \ref{lemma_lower_bound_opt}, the approximation ratio satisfies:
\begin{align}
\label{eq_upper_bound_OPT}
R(w) 
&= \frac{w(S) + w(OPT(\G \setminus S))}{w(OPT(\G))}
 \leq w(S) + w(OPT(\G')) 
 \leq w(S) + \min_{i \in [k]} w(U_i) \nonumber \\
& = y(\delta(S)) + y(E') +   \min_{i \in [k]} y(E_{U_i}) 
= 1 + \min_{i \in [k]} y(E_{U_i}) \leq 1 + \frac{1}{k}. \nonumber
\end{align}
The last equality follows from the fact that $E = E' \cup \delta(S)$ and $y(E) = 1$. The last inequality follows from the fact that the edge sets $\{E_{U_i}\}_{i \in [k]}$ are pairwise disjoint and have a dual sum of at most one, since the total sum of the edges of the graph is $y(E) = 1$. This minimum can thus be upper bounded by $1/k$. 
\end{proof}

In order to prove the upper bound in Theorem \ref{thm_bound_shortest_odd_cycle}, we thus need to construct $\rho$ pairwise edge-separate covers of $\G' = \G \setminus S$. The key for being able to do that is to analyze the structure of the contracted graph $\tilde{\G} = \G /S$, where $S$ is contracted into a single node $v_S$.

\begin{lemma}
\label{thm_separate_covers}
Let $(\G, S)$ be a graph with an odd cycle transversal $S$. If the contracted graph $\tilde{\G}$ contains an odd cycle, then there exists $\rho$ edge-separate feasible covers for the pair $(\tilde{\G}, v_S)$, where $2 \rho - 1$ is the odd girth of $\tilde{\G}$.
\end{lemma}

\begin{proof}
We now dive deeper into the structure of the bipartite graph $\G \setminus S = \tilde{\G} \setminus v_S$. By assumption, this graph admits a bipartition $A \cup B$ of the vertices. Let us assume that it has $k$ connected components $A_1 \cup B_1, \dots, A_k \cup B_k$, all of which are bipartite as well, where $A = \bigcup_i A_i$ and $B = \bigcup_i B_i$. We now fix an arbitrary such component $A_j \cup B_j$.
\begin{itemize}
\item If $v_S$ has an incident edge to both $A_j$ and $B_j$, then this component contains (if including $v_S$) an odd cycle of $\tilde{\G}$. This holds since any path between a node in $A_j$ and a node in $B_j$ has odd length.
\item If $v_S$ has incident edges with only one side, we assume without loss of generality that this side is $A_j$. One could simply switch both sides in the other case while still keeping a valid bipartition of the graph $\tilde{\G} \setminus v_S$.
\item If $v_S$ does not have incident edges with either of the two sides, then $A_j \cup B_j$ is a connected component of $\tilde{\G}$. We call such components \emph{dummy} components and denote by $A_d \cup B_d$ the bipartite graph formed by taking the union of all the dummy components.
\end{itemize}

We denote $N_A(v_S) = N(v_S) \cap A$ and $N_B(v_S) = N(v_S) \cap B$. We now split the graph into layers, where each layer corresponds to the nodes at the same shortest path distance from $N_A(v_S)$. More precisely, we define
\begin{equation}
\label{eq_layer_decomposition}
\mathcal{L}_i := \Big\{v \in A \cup B \mid d(N_A(v_S), v) = i\Big\} \quad \text{for } i \in \{0,\dots,q\}
\end{equation}
where $d(N_A(v_S), v)$ represents the unweighted shortest path distance between $v$ and a vertex in $N_A(v_S)$. The parameter $q$ is defined to be the maximal finite distance from $N_A(v_S)$ in the graph $\tilde{\G}$. An important observation is the fact that these layers are alternatingly included in one side of the bipartition, see Figure \ref{fig_layers} for an illustration of the construction. 

\begin{figure}
\center
\includegraphics[width = 0.7\textwidth]{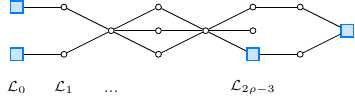}
\caption{The layers of a bipartite graph $\tilde{\G} \setminus v_S = (A \cup B, E')$ with $\rho = 4$. The blue square vertices correspond to $N(v_S)$, where the two left ones are $\mathcal{L}_0 = N_A(v_S)$ and the two right ones are $N_B(v_S)$.}
\label{fig_layers}
\end{figure}

If the graph $\tilde{\G}$ is not connected, note that $d(N_A(v_S),v) = \infty$ for the vertices $v$ lying in dummy components. In order to add the dummy components to the layers and keep alternation between the two sides of the bipartition, we define the last two layers to either be $\{\mathcal{L}_{q+1} := A_d, \: \mathcal{L}_{q+2}:=B_d\}$ or $\{\mathcal{L}_{q+1} := B_d, \: \mathcal{L}_{q+2}:=A_d\}$, depending on which side of the bipartition the last connected layer $\mathcal{L}_{q}$ lies. We now have that $\mathcal{L}_i \subset A$ if $i$ is even, and $\mathcal{L}_i \subset B$ if $i$ is odd. In fact,
\[A = \bigcup_{i = 0}^{\lfloor l/2 \rfloor} \mathcal{L}_{2i} \quad \text{and} \quad B = \bigcup_{i = 1}^{\lceil l/2 \rceil} \mathcal{L}_{2i-1},\]
where the parameter $l \in \mathbb{N}$ represents the index of the last layer: if $\tilde{\G}$ is connected, then $l = q$, otherwise $l = q+2$. Notice also that $\mathcal{L}_0 = N_A(v_S)$. However, $N_B(v_S)$ may now have several different vertices in different layers, see Figure \ref{fig_layers}.

Let $C \subset V$ be an arbitrary odd cycle of $\tilde{\G}$. Notice that this cycle contains $v_S$, a vertex from $N_A(v_S)$ and a vertex from $N_B(v_S)$, since $\tilde{\G} \setminus v_S$ is bipartite and therefore does not contain an odd cycle. Any odd cycle $C$ in $\tilde{\G}$ thus corresponds to an odd path between a vertex in $N_A(v_S) = \mathcal{L}_0$ and a vertex in $N_B(v_S)$. By the assumption that the shortest odd cycle length of $\mathcal{\tilde{G}}$ is $2 \rho -1$, the first layer having a non-empty intersection with $N_B(v_S)$ is $\mathcal{L}_{2\rho -3}$.
A shortest odd cycle of length $2 \rho -1$ therefore corresponds to an odd path of length $2\rho -3$ between $\mathcal{L}_0$ and a vertex in $\mathcal{L}_{2 \rho -3} \cap N_B(v_S)$, see Figure \ref{fig_layers} for an illustration. We now define edges connecting two consecutive layers $\mathcal{L}_i$ and $\mathcal{L}_{i+1}$ as follows:
\[E[\mathcal{L}_i, \mathcal{L}_{i+1}]:= \{(u, v) \in E' \mid u \in \mathcal{L}_i, v \in \mathcal{L}_{i+1}\} \qquad \forall i \in \{0,\dots, l-1\}.\]
We also denote by 
\[\delta_A(v_S) = \{(v_S, u) \in \tilde{E} \mid u \in A\}, \qquad \delta_B(v_S) = \{(v_S, u) \in \tilde{E} \mid u \in B\}\]
the incident edges to $v_S$ respectively connecting to $A$ and $B$.

We are now ready to construct our desired $\rho$ pairwise edge-separate covers of $\tilde{\G}\setminus v_S$, that we denote by $U_1, \dots, U_\rho$ and illustrated in Figure \ref{fig_feasible_opts_g_vp}. Firstly, notice that taking one side of the bipartition is a feasible vertex cover. We thus define $U_{\rho} = A$ and $U_{\rho - 1} = B$. Observe that we then have $E_{U_\rho} = \delta_A(v_S)$ and $E_{U_{\rho -1}} = \delta_B(v_S)$.
\begin{figure}[t]
\center
\includegraphics[width = 0.9\textwidth]{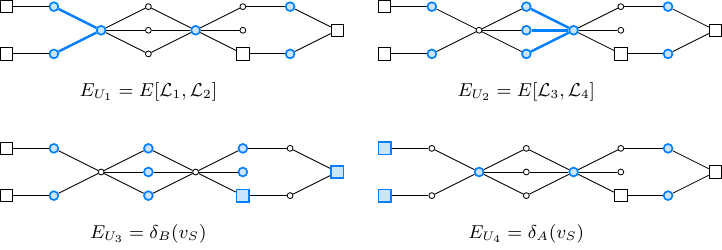}
\caption{The $\rho$ feasible covers of $\G'$ constructed in the proof of Lemma \ref{thm_separate_covers}.}
\label{fig_feasible_opts_g_vp}
\end{figure}
We now construct $\rho - 2$ additional covers with the help of the layers. If $\rho \neq 2$, fix a $j \in [\rho - 2]$, and start the cover $U_j$ by taking the two consecutive layers $\mathcal{L}_{2j-1}$ and $\mathcal{L}_{2j}$. Complete this cover by taking remaining layers alternatingly (hence always skipping one) until covering every edge of the graph. Notice that this cover has an empty intersection with $N(v_S)$. We then have that 
\[E_{U_\rho} = \delta_A(v_S), \qquad E_{U_{\rho -1}} = \delta_B(v_S), \qquad E_{U_j} = E[\mathcal{L}_{2j-1}, \mathcal{L}_{2j}] \quad \forall j \in [\rho -2],\]
which are all pairwise disjoint edge sets, finishing the proof.
\end{proof}

We now have all the tools to prove the upper bound of Theorem \ref{thm_bound_shortest_odd_cycle}.
\begin{proof}
Asume first that $\rho < \infty$, meaning that $\tilde{\G}$ contains an odd cycle. By Lemma \ref{thm_separate_covers}, there exists $\rho$ pairwise edge-separate covers for the pair $(\tilde{\G}, v_S)$. These covers are then still edge-separate for the pair $(\G, S)$, since the bipartite graph is the same in both cases, i.e. $\G' = \tilde{\G}\setminus v_S = \G \setminus S$. This finishes the proof by Lemma \ref{thm_bound_using_pecovers}.

If $\rho = \infty$, then $\tilde{\G}$ is bipartite, with a bipartition $\tilde{A} \cup \tilde{B}$. Assume without loss of generality that $v^S \in \tilde{A}$. Note that $\tilde{E} = E' \cup \delta(v^S)$ and thus $1 = y(\tilde{E}) = y(E') + y(\delta(v^S))$. Any feasible cover of $\G' = \G \setminus S$ needs to count the dual value of every edge in $E'$ at least once. Taking the cover $\tilde{A} \setminus v^S$ counts every edge in $E'$ exactly once, showing that $w(OPT(\G\setminus S)) = y(E')$. Hence, $R(w) \leq w(S) + w(OPT(\G\setminus S)) = y(\delta(v_S)) + y(E') = 1$.
\end{proof}

We now show that this bound is tight and is attained for a class of weight functions $w \in \mathcal{W}$ for any such graph $\G$ and stable set $S$. 

For the case where $\rho = \infty$, it is clear that the approximation ratio always satisfies $R(w) \geq 1$, showing that the bound in Theorem \ref{thm_bound_shortest_odd_cycle} is tight for any $w \in Q^W$. 

We thus assume that $\rho < \infty$. Let $\mathcal{C}$ be all the shortest odd cycles (of length $2 \rho - 1$) of the graph $\tilde{\G}$, each of which is containing $v_S$. For every such cycle $C \in \mathcal{C}$, we define the following dual function on the edges $y^C: \tilde{E} \to \mathbb{R}_{+}$: set both dual edges incident to $v_S$ to $1/\rho$ and then alternatingly set the dual edges to 0 and $1/\rho$ along the odd cycle. For any edge outside of $C$, set its dual value to $0$. Such a solution clearly satisfies $y^C(\tilde{E}) = 1$.
We now take the convex hull of all these functions:
\begin{align*}
\mathcal{Y} &:= \Big\{y: \tilde{E} \to \mathbb{R}_{+} \mid y = \sum_{C \in \mathcal{C}} \lambda^C y^C, \quad \sum_{C \in \mathcal{C}} \lambda^C = 1, \quad  \lambda^C \geq 0 \quad \forall C \in \mathcal{C}\Big\}.
\end{align*}
Because of the one-to-one correspondence between the edge sets $\tilde{E}$ and $E$, due to the fact that $S$ is a stable set, we can naturally define a weight function on the original vertex set once we fix a $y \in \mathcal{Y}$ by setting $w_v := y(\delta(v))$ for every $v \in V$.
We define the space of all such weight functions as
\[\mathcal{W} := \left\{w: V \to \mathbb{R}_{+} \mid w_v = y(\delta(v)) \quad \forall v \in V, \quad y \in \mathcal{Y} \right\}.\]

\begin{figure}[t]
\center
\includegraphics[width = 0.6\textwidth]{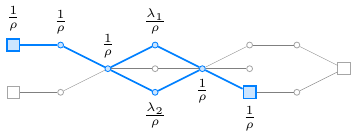}
\caption{An example of a weight function $w \in \mathcal{W}$ obtained by a convex combination of two basic weight functions of shortest odd cycles. }
\label{fig_conv_comb_cycles}
\end{figure}

\begin{theorem}
\label{thm_tightness_bound}
For any weight function $w\in \mathcal{W}$, the approximation ratio satisfies
\[R(w) = 1 + \frac{1}{\rho} \]
where $2 \rho - 1$ is the odd girth of $\mathcal{\tilde{G}}$ and satisfies $\rho \in [2,\infty)$.
\end{theorem}
\begin{proof}
Let $\mathcal{C}$ be the set of all the shortest odd cycles (of length $2 \rho -1$) of the graph $\tilde{\G}$ and let $w \in \mathcal{W}$ with the corresponding $y = \sum_{C \in \mathcal{C}} \lambda^C y^C$. Notice that, for any subset of vertices $U \subset V'$ of the bipartite graph $\G'$, we can count its weight as
\begin{align}
\label{eq_count_weight_through_cycles}
w(U) &= \sum_{v \in U}w_v = \sum_{v \in U}y(\delta(v)) = \sum_{v \in U} \sum_{C \in \mathcal{C}}  \lambda^C y^C(\delta(v)) \nonumber \\ & =  \sum_{v \in U} \sum_{C \in \mathcal{C}}\frac{\lambda^C}{\rho} 1_{\{v \in C\}} = \frac{1}{\rho}\sum_{C \in \mathcal{C}} \lambda^C \sum_{v \in U} 1_{\{v \in C\}} = \frac{1}{\rho}\sum_{C \in \mathcal{C}}  \lambda^C |U \cap C|.
\end{align}
The end of the proof now heavily uses the decomposition of $\tilde{\G}$ into layers described in \eqref{eq_layer_decomposition}. Notice that every odd cycle $C \in \mathcal{C}$ intersects each layer $\mathcal{L}_i$ for $i \in \{0,\dots, 2 \rho - 3\}$ exactly once. Therefore, by \eqref{eq_count_weight_through_cycles},
$w(\mathcal{L}_i) = 1/\rho$ for every $i \in \{0,\dots, 2 \rho - 3\}.$
We now claim that 
\[w(OPT(\G)) = 1, \quad w(OPT(\G')) = \frac{\rho -1}{\rho} \quad \text{and} \quad  w(S) = \frac{2}{\rho}.\]
The fact that $w(OPT(\G)) \geq 1$ follows from Lemma \ref{lemma_lower_bound_opt}. For the reverse inequality, notice that it is possible to take a feasible cover by taking exactly $\rho$ layers in addition to the zero weight vertices, for instance $\mathcal{L}_0 \cup \mathcal{L}_2 \cup \mathcal{L}_3 \cup \mathcal{L}_5 \dots \cup \mathcal{L}_{2 \rho - 3}$, showing $w(OPT(\G)) \leq 1$. 

Observe now that $w(OPT(\G')) = w(OPT(\G \setminus S)) = w(OPT(\tilde{\G}\setminus v_S))$. After removal of $v_S$, every cycle $C \in \mathcal{C}$ becomes a path of length $2 \rho - 3$ (and thus consisting of  $2 \rho -2$ vertices), with one vertex in each layer $\mathcal{L}_i$ for $i \in \{0, \dots, 2 \rho - 3\}$. By feasibility, $OPT(\G')$ has to contain at least $\rho - 1$ vertices for every such path. Using \eqref{eq_count_weight_through_cycles}, we infer $w(OPT(\G')) \geq (\rho - 1)/\rho$. For the reverse inequality, taking $\rho -1$ layers alternatively, such as $\mathcal{L}_0 \cup \mathcal{L}_2 \cup \mathcal{L}_4 \dots \cup \mathcal{L}_{2 \rho - 4}$, as well as the zero weight vertices, builds a feasible cover of weight exactly $(\rho - 1)/\rho$. 

Finally, notice that
$w(S) = w(v_S) = 2/\rho$
because every $C \in \mathcal{C}$ contains $v_S$. 
By combining the three equalities, we get the desired result
\[R(w) = \frac{w(S) + w(OPT(\G'))}{w(OPT(\G))} = 1 + \frac{1}{\rho}.\]
\end{proof}

\subsection{Arbitrary Set to Bipartite}
\label{section_arb_set_to_bipartite}
We now consider the setting where $S$ is now an arbitrary set. Our guarantee on the approximation ratio will now also depend on the total sum of the dual variables on the edges inside of the set $S$. We denote this sum by
\[\alpha := y(E[S]) \in [0,1].\]

\begin{figure}
\center
\includegraphics[width = 0.7\textwidth]{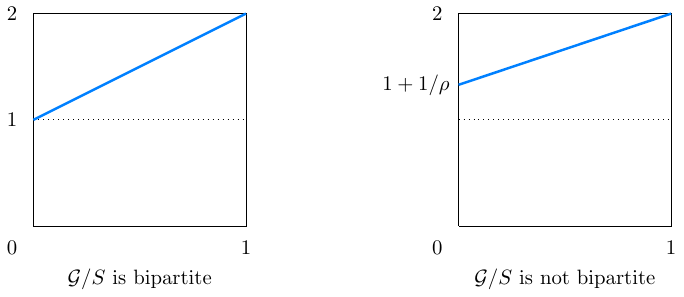}
\caption{The plot of the approximation ratio $R(w)$ with respect to the parameter $\alpha \in [0,1]$.}
\label{fig_approx_ratio}
\end{figure}

\begin{theorem}
\label{thm_bound_ratio_general}
For any $w \in \Qw$, the approximation ratio satisfies 
\[R(w) \leq \left(1+\frac{1}{\rho}\right) (1 - \alpha) + 2 \alpha \qquad \text{with } \alpha \in [0,1] \text{ and } \rho \in [2,\infty]\]
where $2\rho -1$ denotes the odd girth of the contracted graph $\tilde{\G}$. Moreover, these bounds are tight and are attained for any $\alpha \in [0,1]$ and any $\rho \in [2,\infty]$.
\end{theorem}
\begin{proof}
We only prove here the upper bound with $\rho < \infty$ and leave the remaining statements to the Appendix. The proof essentially follows the same arguments as the one of Lemma \ref{thm_bound_using_pecovers} with the $\alpha$ parameter incorporated, and we thus only highlight the main differences. We decompose the weight of the set $S$ with respect to the dual variables. The edges in $E[S]$ are counted twice, whereas the edges in $\delta(S)$ are counted once:
\begin{equation}
\label{eq_arb_set_weight_S}
w(S) = 2\alpha + y(\delta(S)).
\end{equation}
Consider the contracted graph $\tilde{\G} = \G / S = (\tilde{V}, \tilde{E})$ and denote by $v^S$ the contracted node. The edge set of this graph is now
$\tilde{E} = \delta(S) \cup E'$
, since the edges in $E[S]$ have been collapsed. By Lemma \ref{thm_separate_covers}, we can construct $\rho$ edge-separate covers $U_1, \dots, U_\rho \subset \tilde{V} \setminus v^S$ for the pair $(\tilde{\G},v_S)$. These covers are still edge-separate for the pair $(\G, S)$, implying
\begin{equation}
\label{eq_arb_set_opt_G_setminus_S}
w(OPT(\G \setminus S)) \leq \min_{i \in [\rho]} \; w(U_i) = y(E') + \min_{i \in [\rho]} y(E_{U_i}) \leq y(E') + \frac{1-\alpha}{\rho}.
\end{equation}
The first inequality holds since every $U_i$ is a feasible cover of $\G \setminus S$. The second equality holds by counting the weight of a cover $U_i$ in terms of the dual edges. The last inequality holds because the edge sets $\{E_{U_i}\}_{i \in [\rho]}$ are pairwise disjoint, and their total dual sum is at most $1 - \alpha$.
Combining Lemma \ref{lemma_lower_bound_opt}, \eqref{eq_arb_set_weight_S} and \eqref{eq_arb_set_opt_G_setminus_S},
\begin{align*}
R(w) 
\leq 2 \alpha + y(\delta(S)) + y(E') + \frac{1-\alpha}{\rho} 
= 1 + \alpha + \frac{1-\alpha}{\rho} = \left(1+\frac{1}{\rho}\right) (1 - \alpha) + 2 \alpha.
\end{align*}
\end{proof} 

\section{Algorithmic applications}
\label{section_alg_consequences}
We focus in this section on efficient ways to find odd cycle transversals with a low value for the $\alpha$ parameter. In fact, once such a set $S$ is found, there can also be freedom in the choice of the dual solution in order to optimize the $\alpha$ parameter. This motivates the following definition.\

\begin{definition}
Let $(\G, S, y, w)$ be a graph with an odd cycle transversal  $S \subset V$, weights $w \in Q^W$ and a dual solution $y \in \mathbb{R}^E_+$. A tuple $(\G', S', y', w')$ is \emph{approximation preserving} if
\[w(S) + w(OPT(\G \setminus S)) \leq w'(S') + w'(OPT(\G' \setminus S')).\]
Moreover, we say that $\alpha \in [0,1]$ is \emph{valid} for the pair $(\G, S)$ if there exists an approximation preserving $(\G', S', y', w')$ such that $\alpha = y'(E[S']).$
\end{definition}
Finding a valid $\alpha \in [0,1]$ would directly allow us to use it in the bound of Theorem \ref{thm_bound_ratio_general}, where the $\rho$ parameter would correspond to the one of the approximation preserving graph. We present here an application if a coloring of a graph can be found efficiently.

\begin{theorem} 
\label{thm_coloring_application}
Let $\G = (V,E)$ be a graph with weights $w \in Q^W$ that can be $k$-colored in polynomial time for $k \geq 4$. There exists an efficiently findable set $S \subset V$ bipartizing the graph and a valid $\alpha$ such that 
$ \alpha  \leq 1 - 4/k$, leading to an approximation ratio of
\[R(w) \leq 2 - \frac{4}{k}\left(1 - \frac{1}{\rho} \right)\]
\end{theorem}
\label{thm_alpha_bound}
\begin{proof}
Let us denote by $V_1, \dots, V_k$ the $k$ independent sets defining the color classes of the graph $\G$. We assume without loss of generality that they are ordered by weight $w(V_1) \leq w(V_2) \dots \leq w(V_k)$. Since $w(V) = 2$, the two color classes with the largest weights satisfy $w(V_{k-1}) + w(V_k) \geq 4/k$. We define the bipartizing set to be the remaining color classes: $S := V_1 \cup \dots \cup V_{k-2}$. We denote by $y \in \mathbb{R}^E_+$ the dual solution satisfying complementary slackness and $y(E) = 1$.

We now define an approximation preserving $(\G', S', y', w')$ in the following way. Let $\G' = K_k$ be the complete graph on $k$ vertices, denoted by $\{v_1, \dots, v_k\}$. The weights are defined to be \[w'(v_i) := w(V_i) \quad \text{ and } \quad y'(v_i, v_j) := y(E[V_i, V_j])\] for every $i,j \in [k]$. These clearly satisfy the complementary slackness condition $y'(\delta(v_i)) = w'(v_i)$ for every $i \in [k]$. The bipartizing set is defined to be $S' := \{v_1, \dots, v_{k-2}\}$. This tuple is approximation preserving since $w(S) = w'(S')$ and $w(OPT(\G \setminus S)) \leq w'(OPT(\G'\setminus S'))$. In order to prove the theorem, we still need to tweak the dual solution $y'$ to ensure $\alpha := y'(E[S']) \leq 1 - 4/k$. Observe that $w'(v_{k-1}) + w'(v_k) \geq 4/k$.

\begin{enumerate}
\item If $y'(v_{k-1},v_k) = 0$, then the result follows since in that case $y'(\delta(S')) \geq 4/k$ and thus $y'(E[S']) \leq 1 - 4/k$. 
\item If $y'(E[S']) = 0$, then the result trivially follows as well.
\end{enumerate}

Suppose thus that $y'(E[S']) > 0$ and $y'(v_{k-1}, v_k) > 0$. Pick an arbitrary edge $(v_i, v_j) \in E[S']$ satisfying $y'(v_i, v_j) > 0$ and consider the 4-cycle $(v_i, v_j, v_{k-1}, v_k)$. Notice that alternatively increasing and decreasing the dual values on the edges of this cycle by a small amount $\epsilon > 0$ gives another feasible dual solution satisfying the complementary slackness condition. More formally, we set $\epsilon := \min\{y'(v_i, v_j), y'(v_{k-1}, v_k)\}$, decrease $y'(v_i, v_j)$ and $y'(v_{k-1}, v_k)$ by $\epsilon$, while increasing $y'(v_j, v_{k-1})$ and $y'(v_k, v_i)$ by the same amount. Observe that this leads to either $(v_i, v_j)$ or $(v_{k-1}, v_k)$ dropping to dual value zero. We can repeat this procedure until either $y'(E[S']) = 0$ or $y'(v_{k-1}, v_k) = 0$, finishing the proof of the theorem.
\end{proof}

We now claim that this result is optimal in the following sense. Consider an $n$-vertex graph. It is known that the integrality gap of the standard linear programming relaxation for vertex cover is upper bounded by $2 - 2/n$, a bound which is attained on the complete graph. This implies that any approximation algorithm lower bounding $w(OPT)$ by comparing it to the optimal LP solution, as we do in Lemma \ref{lemma_lower_bound_opt}, cannot do better than $2 - 2/n$ in the worst case. Setting $\rho = 2$ in Theorem \ref{thm_coloring_application}, which corresponds to the worst case since $\rho \in [2,\infty]$, recovers this bound and a result of Hochbaum in \cite{hochbaum1983efficient}.

We now make another connection with the \Call{MinUncut}{} problem, which is defined on a graph $\G = (V, E)$ with weights $y_e$ for every edge $e \in E$. The goal of the problem is to find a cut of the graph which minimizes the total weight of uncut edges. This problem is NP-hard and admits a $O(\sqrt{\log(n)})$ approximation \cite{Agarwal2005OlogNA}. We call a \Call{MinUncut}{} instance \emph{light} if its optimal solution is bounded above by $y(E)/\Omega(\log(n))$.
\begin{theorem}
For any light \Call{MinUncut}{} instance, combining the $O(\sqrt{\log(n)})$ approximation in \cite{Agarwal2005OlogNA} with Algorithm \ref{round_bipartize} outputs a vertex cover with approximation ratio at most
\[R(w) \leq 1 + \frac{1}{\rho} + o(1)\]
where $2 \rho -1$ is the odd girth of the contracted graph $\tilde{\G}$ and satisfies $\rho \in [2,\infty]$.
\end{theorem}
\begin{proof}
The key observation is that any feasible solution to the \Call{MinUncut}{} problem of value $\alpha$ gives an odd cycle transversal $S$ satisfying $y(E[S]) = \alpha$. Indeed, let $(C, V \setminus C)$ be a feasible solution, where the total weight of uncut edges is $\alpha$. Observe that removing all the uncut edges makes $C$ and $V \setminus C$ become stable sets, implying that the remaining graph is bipartite. We then define $S \subset V$ to be all the nodes incident to the uncut edges. Since removing all the nodes in $S$ also removes all the uncut edges, the remaining graph is bipartite. Moreover, $y(E[S]) = \alpha$. 

By the lightness assumption and the weight space normalization, the optimal solution $\alpha^*$ satisfies $\alpha^* \leq y(E) / \Omega(\log(n)) = 1 / \Omega(\log(n))$. Running the $O(\sqrt{\log(n)})$ approximation algorithm then outputs a solution with value $\alpha \leq O(\sqrt{\log(n)}) \; \alpha^* \leq 1 / \Omega(\sqrt{\log(n)}) = o(1)$. This therefore leads to an approximation guarantee of 
\[R(w) = \left(1+\frac{1}{\rho}\right) (1 - \alpha) + 2 \alpha = 1 + \frac{1}{\rho} + \alpha \left(1 - \frac{1}{\rho}\right) \leq 1 + \frac{1}{\rho} + o(1).\]
\end{proof}

\section{Integrality Gap and Fractional Chromatic Number}
\label{section_integrality_gap}

We focus in this section on proving tight bounds for the integrality gap of $3$-colorable graphs. A key result that we use in this section is given by Singh in \cite{singh2019integrality}, which relates the integrality gap with the fractional chromatic number of a graph. The latter is denoted as $\chi^f(\G)$ and is defined as the optimal solution of the following primal-dual linear programming pair. We denote by $\mathcal{I} \subset 2^V$ the set of all independent sets of the graph $\G$. Solving these linear programs is however NP-hard because of the possible exponential number of independent sets.

\begin{minipage}[t]{0.5\textwidth}
\begin{align*}
    \min \sum_{I \in \mathcal{I}} y_I  & \\
    \sum_{I \in \mathcal{I}, v \in I} y_I &\geq 1 \quad \forall v \in V\\
    y_I &\geq 0 \quad \forall I\in \mathcal{I}
\end{align*}
\end{minipage}
\begin{minipage}[t]{0.5\textwidth}
\begin{align*}
    \max \sum_{v \in V} z_v  & \\
    \sum_{v \in I} z_v &\leq 1 \quad \forall I \in \mathcal{I}\\
    z_v &\geq 0 \quad \forall v\in V
\end{align*}
\end{minipage}
\vspace{0.5cm}

\noindent Note that $\chi^f(\G) = 2$ if and only if $\G$ is bipartite.

\begin{theorem}[Singh, \cite{singh2019integrality}]
\label{thm_singh_integrality_gap}
Let $\G = (V,E)$ be a graph. The integrality gap of the vertex cover linear programming relaxation $P(\G)$ satisfies \[IG(\G) = 2 - \frac{2}{\chi^f(\G)}.\]
\end{theorem}

We first focus on graphs with the existence of a single vertex whose removal produces a bipartite graph. The following theorem generalizes the result given for the cycle graph in \cite{arora2002proving, scheinerman2011fractional} and turns out to be the same formula as for series-parallel graphs \cite{goddard2016fractional}.
\begin{theorem}
\label{thm_fractional_chrom_number}
Let $\G = (V,E)$ be a non-bipartite graph and $v_p \in V$ such that $\G \setminus v_p = (A \cup B, E')$ is bipartite. Then,
\[\chi^f(\G) = 2 + \frac{1}{\rho - 1},\]
where $2 \rho - 1$ is the odd girth of $\G$.
\end{theorem}

\begin{proof}[Proof of Theorem \ref{thm_fractional_chrom_number}]
We prove this theorem by constructing feasible primal and dual solutions of objective value $2 + 1/(\rho -1)$. By strong duality, these two solutions are then optimal for their respective linear programs, hence proving the theorem.

We first construct the dual solution. Let $\mathcal{C}$ be the set of all the shortest odd cycles of $\G$. For any such cycle $C \in \mathcal{C}$, define the dual solution $z^C \in \mathbb{R}^V$ by
\[z_v^C = \begin{cases} 1/(\rho -1) \quad &\text{if } v \in C \\ 0 \quad &\text{if } v \in V \setminus C \end{cases}\]
This solution is feasible since any independent set in an odd cycle of length $2\rho -1$ has size at most $\rho -1$. Indeed, fix an independent set $I \in \mathcal{I}$, then:
\[\sum_{v \in I} z_v^C = \sum_{v \in I \cap C}\;  \frac{1}{\rho -1} = \frac{|I \cap C|}{\rho -1} \leq 1.\]
Moreover, the objective value of this solution is:
\[\sum_{v \in V} z_v^C = \sum_{v \in C} \; \frac{1}{\rho -1} = \frac{2 \rho -1}{\rho -1} = 2 + \frac{1}{\rho -1}.\]

\begin{figure}
\centering
\includegraphics[width = 0.5\textwidth]{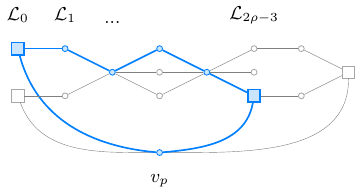}
\caption{An optimal dual solution constructed in the proof of Theorem \ref{thm_fractional_chrom_number}. Each node on a shortest odd cycle is assigned a fractional value of $1/(\rho -1)$.}
\label{fig_FCN_primal_dual_solutions}
\end{figure}

Let us now construct the primal solution. We will do so by constructing $2\rho -1$ independent sets $I_k \in \mathcal{I}$ and assigning to each of them a fractional value of $y(I_k) = 1/(\rho -1)$. All the other independent sets are assigned value zero. We split the bipartite graph $\G \setminus v_p$ into the layers
\[\mathcal{L}_i := \{v \in A \cup B \mid d(N_A(v_p), v) = i\} \quad \text{for } i \in \{0,\dots,l\}.\] as explained in \eqref{eq_layer_decomposition}. As a reminder, any shortest odd cycle corresponds to a path between $\mathcal{L}_0 = N_A(v_p)$ and $\mathcal{L}_{2 \rho -3} \cap N_B(v_p)$. The original vertex set $V$ is thus decomposed into $\{v_p\} \cup \mathcal{L}_0 \cup \dots \cup \mathcal{L}_l$, where each layer is an independent set and only has edges going out to $v_p$ or its two neighbouring layers. 

Let us first focus on the subgraph consisting of the vertices in $\{v_p\} \cup \bigcup_{i = 0}^{2 \rho -3} \mathcal{L}_i$, where any shortest odd cycle has exactly one vertex per layer (per abuse of notation, we say that $\{v_p\}$ is also a layer in this situation). For convenience of indexing, we rename these layers as $\mathcal{\tilde{L}}_1, \dots, \mathcal{\tilde{L}}_{2 \rho -1}$ where $\tilde{\mathcal{L}_1} = v_p$ and  $\tilde{\mathcal{L}}_i = \mathcal{L}_{i-2}$ for $i > 1$. We now create $2 \rho -1$ independent sets on this subgraph in the following way. The first independent set is defined as $U_1 = \mathcal{\tilde{L}}_1 \cup \mathcal{\tilde{L}}_4 \cup \mathcal{\tilde{L}}_6\ \dots \cup \mathcal{\tilde{L}}_{2 \rho -2}$, where we take the first layer $\mathcal{\tilde{L}}_1$, skip two before taking the next one and then continue by taking the remaining layers alternatingly (hence always skipping one), see Figure \ref{fig_FCN_primal_solution}. Note that the layer following $\mathcal{\tilde{L}}_{2 \rho -1}$ is assumed to be $\mathcal{\tilde{L}}_1$. This procedure generates in fact a distinct independent set by starting at $\mathcal{\tilde{L}}_k$ for any $k \in [2 \rho -1]$ and we denote the corresponding independent set by $U_k$. Notice that each layer is contained in exactly $\rho -1$ of the constructed independent sets $\{U_k \mid k \in [2 \rho -1]\}$.

\begin{figure}
\centering
\includegraphics[width = 0.8\textwidth]{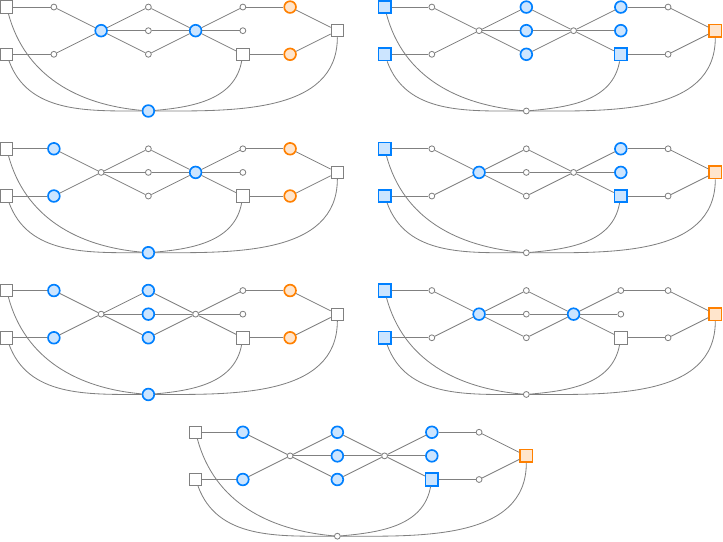}
\caption{The $2\rho -1$ independent sets $I_k$ constructed in the optimal primal solution. The blue nodes correspond to $\{U_k \mid k \in [2 \rho -1]\}$, whereas the orange nodes correspond to $\{R_1, R_2\}.$}
\label{fig_FCN_primal_solution}
\end{figure}

We now focus on the subgraph consisting of the vertices in $\bigcup_{i > 2 \rho -3} \mathcal{L}_i$. We can construct two different independent sets there by taking either the odd or even indexed layers, i.e. 
\[R_{1} := \bigcup_{i \text{ odd}, \; i > 2 \rho -3} \mathcal{L}_i \quad \text{and} \quad R_{2} := \bigcup_{i \text{ even}, \; i > 2 \rho -3} \mathcal{L}_i. \]
We now define our final $2 \rho -1$ independent sets on the full graph as:
\[I_k := \begin{cases} U_k \cup R_1 \quad \text{ if } v_p \notin U_k \\ U_k \cup R_2 \quad \text{ if } v_p \in U_k \end{cases} \qquad \forall k \in [2 \rho - 1]. \]
These are in fact independent sets: in the first case, the first layer in $R_1$ is $\mathcal{L}_{2 \rho -1}$ whereas the last layer in $U_k$ has index at most $2 \rho -3$, meaning that there are no two neighbouring layers. In the second case, since $v_p \in U_k$, we have that $\mathcal{L}_{2 \rho -3} \notin U_k$, by construction of $U_k$. The last layer in $U_k$ thus has index at most $2 \rho -4$, whereas the first layer in $R_2$ is $\mathcal{L}_{2 \rho -2}$, meaning again that there are no two neighbouring layers. In addition, there is no edge between $v_p$ and $R_2$, because the only even indexed layer having edges sent to $v_p$ is $\mathcal{L}_0 = N_A(v_p)$.

We now define our primal solution as
\[y(I_k) = \frac{1}{\rho -1} \qquad \forall k \in [2 \rho -1],\]
and $y(I) = 0$ for every other independent set $I \in \mathcal{I}$. We now show this is a feasible solution, i.e. that every vertex $v \in V$ belongs to at least $\rho -1$ independent sets in $\{I_k \mid k \in [2\rho -1]\}$. For $v \in \{v_p\} \cup \bigcup_{i = 0}^{2 \rho -3} \mathcal{L}_i$, such a vertex lies by construction in exactly $\rho -1$ independent sets $\{U_k \mid k \in [2 \rho -1]\}$, and thus also of $\{I_k \mid k \in [2 \rho -1]\}$. For $v \in \bigcup_{i > 2 \rho -3} \mathcal{L}_i$, if $v$ belongs to an even indexed layer, then it is contained in $\rho -1$ of the desired independent sets. If it belongs to an odd indexed layer, then it is contained in $\rho$ of them. Therefore,
\[\sum_{I \in \mathcal{I}, v \in I} y_I = \sum_{k = 1}^{2 \rho -1} y(I_k) \; 1_{\{v \in I_k\}} = \frac{1}{\rho - 1} \sum_{k = 1}^{2 \rho -1} 1_{\{v \in I_k\}} \geq 1.\]
The objective value of this primal solution is clearly $2 + 1/(\rho -1)$. We have constructed feasible primal and dual solutions with the same objective value. By strong duality, this finishes the proof of the theorem.

\end{proof}

We now consider the case where $\G = (V,E)$ is a graph with chromatic number $\chi(\G) = 3$.

\begin{theorem}
\label{thm_upper_bound_FCN}
Let $\G = (V,E)$ be a 3-colorable graph with color classes $V = V_1 \cup V_2 \cup V_3$. Then,
\[\chi^f(\G) \leq 2 + \min_{i \in \{1,2,3\}} \frac{1}{\rho_i - 1}\]
where $2 \rho_i -1$ is the odd girth of the contracted graph $\G/V_i$ for each $i \in \{1,2,3\}$. Moreover, equality holds if one color class only contains one vertex.
\end{theorem}
\begin{proof}
We prove this theorem by constructing three feasible solutions of value $2 + 1/(\rho_i -1)$ for each $i \in \{1,2,3\}$ to the primal linear program of the fractional chromatic number on the graph $\G$. 

Fix an $i \in \{1,2,3\}$ and consider the graph $\tilde{\G}:= \G / V_i = (\tilde{V}, \tilde{E})$ with odd girth $2 \rho_i -1$. We denote the contracted node by $\tilde{v} \in \tilde{V}$. Since this graph is bipartite if we were to remove $\tilde{v}$, we know that its fractional chromatic number is equal to $2 + 1/(\rho_i -1)$ by Theorem \ref{thm_fractional_chrom_number}. Let $\{\tilde{I}_k, k \in [2\rho_i -1]\}$ be the independent sets in the support of the optimal primal solution of the graph $\tilde{\G}$ constructed in the proof of this theorem. For each of these independent sets, we extend them to the original graph in the following way:
\[I_k = \begin{cases} \tilde{I}_k \qquad &\text{if } \tilde{v} \notin \tilde{I}_k \\ (\tilde{I}_k \setminus \tilde{v}) \cup V_i \qquad &\text{if } \tilde{v} \in \tilde{I}_k. \end{cases}\]
In words, if $\tilde{v}$ happens to belong to $\tilde{I}_k$, we replace it by $V_i$ to get a valid independent set in the original graph. Assigning fractional value $y(I_k) = 1/(\rho_i -1)$ for every $k \in [2 \rho_i -1]$ yields a feasible primal solution with objective value $2 + 1/(\rho_i -1)$. Since we can do this for every $i \in \{1,2,3\}$, and the optimal minimum value of the primal linear program is at most the objective value of any of these feasible solutions, the proof is finished.

Moreover, this upper bound is in fact tight, since it holds with equality when one of the color classes only contains one vertex by Theorem \ref{thm_fractional_chrom_number}.
\end{proof}

It is now straightforward to extend this result for the integrality gap by Theorem \ref{thm_singh_integrality_gap}.
\begin{corollary}
Let $\G = (V,E)$ be a 3-colorable graph with color classes $V = V_1 \cup V_2 \cup V_3$. The integrality gap $IG(\G)$ of the standard linear programming relaxation $P(\G)$ satisfies:
\[IG(\G) \leq 1 + \min_{i \in \{1,2,3\}} \frac{1}{2 \rho_i - 1}\]
where $2 \rho_i -1$ is the odd girth of the contracted graph $\G/V_i$ for each $i \in \{1,2,3\}$. Moreover, equality holds if one color class only contains one vertex.
\end{corollary}

\bibliographystyle{plain}
\bibliography{references}

\appendix

\section{Arbitrary Set to Bipartite: Omitted Proofs}

\begin{theorem}
Let $\G = (V,E)$ be a graph and $S \subset V$ such that $\G \setminus S = (A \cup B, E')$ is bipartite. For any $w \in \Qw$, the approximation ratio satisfies 
\[R(w) \leq 1 + \alpha \qquad \text{with } \alpha \in [0,1]\]
if the contracted graph $\G / S$ is bipartite.
\end{theorem}
\begin{proof}
The only change with respect to the previous proof is the bound on $w(OPT(\G \setminus S))$ in \eqref{eq_arb_set_opt_G_setminus_S}. We denote the contracted graph by $\tilde{\G} = \G / S$ and by $v^S$ the contracted vertex. Suppose $\tilde{\G}$ admits the bipartition $(\tilde{A} \cup \tilde{B}, \tilde{E})$ and assume without loss of generality that $v^S \in \tilde{A}$. Note that $\tilde{E} = E' \cup \delta(v^S)$.

Any feasible cover of $\G \setminus S$ needs to count the dual value of every edge in $E'$ at least once. Taking the cover $\tilde{A} \setminus v^S$ counts every edge in $E'$ exactly once, showing that $w(OPT(\G\setminus S)) = y(E')$. Hence, using $y(E) = \alpha + y(\delta(S)) + y(E') = 1$, we get
\[R(w) \leq 2 \alpha + y(\delta(S)) + y(E') = 1 + \alpha.\]
\end{proof}

\begin{theorem}
Let $\alpha \in [0,1]$ and let $\rho \geq 2$. There exists a non-bipartite graph $\G = (V,E)$, with weights $(y,w) \in \Qyw$, and a set $S \subset V$ with $y(E[S]) = \alpha$, where $\G/S$ has odd girth $2 \rho -1$ and which satisfies
\[R(w) = \left(1+\frac{1}{\rho}\right) (1 - \alpha) + 2 \alpha.\]
\end{theorem}
\begin{proof}
An example of such a graph can be constructed as follows. We first construct $\G/S$: take an odd cycle of length $2 \rho -1$ with a distinguished node $v^S$ and assign dual value $(1-\alpha)/\rho$ to both edges incident to it. Alternatively assign dual values $0$ and $(1-\alpha)/\rho$ along the odd cycle for the remaining edges. In order to construct $\G$, replace $v^S$ by a triangle $S$ with dual edges set to $\alpha, 0$ and $0$, where the two previous incident edges to $v^S$ are adjacent to the endpoints of the edge with value $\alpha$. Note that we replace it with a triangle instead of a single edge in order to avoid $\G$ becoming bipartite. Similarly to the proof of Theorem \ref{thm_tightness_bound}, one can check that
\[w(S) = 2 \alpha + \frac{2\: (1- \alpha)}{\rho} ; \quad w(OPT(\G \setminus S)) = \frac{(1-\alpha)(\rho -1)}{\rho}; \quad w(OPT(\G)) = 1.\]
Therefore,
\[R(w) = 2 \alpha + \frac{2\: (1- \alpha)}{\rho} + \frac{(1-\alpha)(\rho -1)}{\rho} = \left(1+\frac{1}{\rho}\right) (1 - \alpha) + 2 \alpha.\]
\end{proof}
\begin{theorem}
Let $\alpha \in [0,1]$. There exists a non-bipartite graph $\G = (V,E)$, with weights $(y,w) \in \Qyw$, and a set $S \subset V$ with $y(E[S]) = \alpha$, where $\G/S$ is bipartite and which satisfies
\[R(w) = 1 + \alpha.\]
\end{theorem}
\begin{proof}
Let $\G$ be an arbitrary odd cycle. Consider an arbitrary edge $(u,v) \in E$ and assign it dual value $\alpha$. The set $S$ is defined to be $S = \{u,v\}$. Assign dual value zero to the edge $(u,w) \in E$, where $w$ is the second neighbour of $u$ in the cycle. For the remaining edges, arbitrarily assign dual values, while ensuring that they sum up to $1-\alpha$. The fact that one edge is equal to zero is necessary in order to get the exact formula $w(OPT(\G)) = 1$, a feasible cover showing $w(OPT(\G)) \leq 1$ being the following: take both endpoints of the edge $(u,w)$ and take remaining vertices alternatively (hence always skipping one) along the odd cycle. All the edges are counted once, except for $(u,w)$, which is counted twice but has value zero. Moreover, $w(S) = 2 \alpha + y(\delta(S))$ and $w(OPT(\G \setminus S)) = y(E')$, where $E'$ is the edge set of the bipartite graph $\G \setminus S$. Therefore,
\[R(w) = 2 \alpha + y(\delta(S)) + y(E') = 1 + \alpha.\]
\end{proof}

\section{Weight Space Polytope: Vertices}
We denote by $\mathds{1}_v \in \mathbb{R}^V$ the indicator vector of a vertex $v\in V$.

\begin{theorem}
The polytope $\Qw$ of a graph $\G = (V,E)$ can be expressed as:
\[\Qw = \text{\emph{conv}}\Big(\Big\{\mathds{1}_u + \mathds{1}_v \; \big| \; (u,v) \in E\Big\}\Big).\]
Moreover, $\mathds{1}_u + \mathds{1}_v$ is an extreme point of $\Qw$ for every edge $(u,v) \in E$.
\end{theorem}
\begin{proof}
Let $w \in \Qw$ and $y \in \mathbb{R}^E_+$ such that $y(\delta(v)) = w_v$ for every $v \in V$, as well as $y(E) = 1$. Observe that:
\[w = \sum_{(u,v) \in E} y_{(u,v)} \big(\mathds{1}_u + \mathds{1}_v \big).\]
By looking at this equality coordinate by coordinate, for a fixed vertex $v \in V$, the contribution from the left hand side is $w_v$ whereas the contribution from the right hand side is $y(\delta(v))$. We have thus managed to decompose an arbitrary $w \in \Qw$ into a convex combination of the vectors $\Big\{\mathds{1}_u + \mathds{1}_v \; \big| \; (u,v) \in E\Big\}$.

Let $(u,v) \in E$ and let us now show that $\bar{w} := \mathds{1}_u + \mathds{1}_v$ is an extreme point of $\Qw$. Firstly, it is clear that $\bar{w} \in \Qw$, the dual solution satisfying complementary slackness being $y_{(u,v)} = 1$ and $y_e = 0$ for every $e \in E \setminus (u,v)$, and the sum of all the weights being indeed equal to two. Suppose for contradiction that there exist distinct $w^1, w^2 \in \Qw$ such that $\bar{w} = \frac{1}{2}(w^1 + w^2)$. Notice that, since the weight vectors are required to be non-negative, $\bar{w}_z = w^1_z = w^2_z = 0$ for every $z \in V \setminus \{u,v\}$. Thus, we can assume without loss of generality that $w^1_u = w^2_v = 1 + \epsilon$ and  $w^1_v = w^2_u = 1 - \epsilon$ for some $\epsilon > 0$. However, the fully half-integral solution is now not an optimal LP solution on the weights $w^1$ and $w^2$. Indeed, on the weight function $w^1$, the feasible solution $V \setminus \{u\}$ has objective value $1-\epsilon$, whereas the fully half-integral solution has weight $1$. Similarly, on the weight function $w^2$, the feasible solution $V \setminus v$ has objective value $1-\epsilon$. By Lemma \ref{lemma_feasible_weights}, $w^1, w^2 \notin \Qw$, leading to a contradiction. The weight function $\bar{w}$ thus cannot be written as a non-trivial convex combination of two other points in the polytope $\Qw$ and is therefore an extreme point (or a vertex) of $\Qw$.
\end{proof}

\end{document}